%% file: arxiv.tex
\documentclass[twoside,leqno,twocolumn,10pt]{article}  
\usepackage{ltexpprt} 

\input{part-macros}

\usepackage{chngcntr}

\newcommand*{\QED}{\hfill\ensuremath{\blacksquare}}%

\newcommand{\enablemovecache}{false}
\newcommand{\setcache}[2]{
\ifthenelse{\equal{\enablemovecache}{true}}{%
	\newcommand{#1}{#2}
}{#2} 
}
\newcommand{\putcache}[1]{
\ifthenelse{\equal{\enablemovecache}{true}}{%
	#1
}{} 
}

\begin{document}

\title{Real-time Topic-aware Influence Maximization Using Preprocessing}
\author{Wei Chen\thanks{Microsoft Research. Email: {weic@microsoft.com}.} \\
\and
Tian Lin\thanks{Tsinghua University. Email: {lint10@mails.tsinghua.edu.cn}.}\\
\and
Cheng Yang\thanks{Tsinghua University. Email: {albertyang33@gmail.com}.}}

\date{}
\maketitle

\input{part-abstract}


\section{Introduction}

In a social network, information, ideas, rumors, and innovations can be propagated to a large
	number of people because of the social influence between the connected peers in the network.
{\em Influence maximization} is the task of finding a set of {\em seed nodes} in a social network such that
	the influence propagated from the seed nodes can reach the largest number of people in the network.
More technically, a social network is modeled as a graph with nodes representing individuals and directed
	edges representing influence relationships.
The network is associated with a stochastic diffusion model (such as
	independent cascade model and linear threshold model \cite{kempe2003maximizing}) characterizing the influence propagation
	dynamics starting from the seed nodes.
Influence maximization is to find a set of $k$ seed nodes in the network such that the 
	{\em influence spread}, defined as the expected number
	of nodes influenced (or activated) through influence diffusion starting from the seed nodes,
	is maximized (\cite{kempe2003maximizing,CLC13}).

Influence maximization has a wide range of applications including
	viral marketing~\cite{domingos2001mining,richardson2002mining,kempe2003maximizing},
	information monitoring and outbreak detection~\cite{leskovec2007cost}, 
	competitive viral marketing and rumor control \cite{BAA11,HeSCJ12}, or even text summarization~\cite{WangYLZH13}
	(by modeling a word influence network).
As a result, influence maximization has been extensively studied in the past decade.
Research directions include 
	improvements in the efficiency and scalability of influence maximization 
	algorithms (e.g., \cite{chen2009efficient,wang2012scalable,goyal2011simpath}),
	extensions to other diffusion models and optimization problems
	(e.g., \cite{BAA11,bhagat_2012_maximizing,HeSCJ12}), and influence model
	learning from real-world data (e.g., \cite{saito2008prediction,tang2009social,goyal2010learning}).

Most of these works treat diffusions of all information, rumors, ideas, etc.
	(collectively referred as {\em items} in this paper) as following the same model with a single
	set of parameters.
In reality, however, influence between a pair of friends may differ depending on the topic.
For example, one may be more influential to the other on high-tech gadgets, while the other is more
	influential on fashion topics, or one researcher is more influential on data mining topics to her
	peers but less influential on algorithm and theory topics.
Recently, Barbieri et al. \cite{barbieri2013topic} propose the topic-aware independent cascade (TIC)
	and linear threshold (TLT) models, in which a diffusion item is a mixture of topics and 
	influence parameters for each item are also mixtures of parameters for individual topics.
They provide learning methods to learn influence parameters in the topic-aware models from real-world
	data.
Such topic-mixing models require new thinking in terms of the influence maximization task, which
	is what we address in this paper.

In this paper, we adopt the models proposed in \cite{barbieri2013topic} and 
	study efficient topic-aware influence maximization schemes.
One can still apply topic-oblivious influence maximization algorithms in online processing
	of every diffusion item, but it may not be efficient when there are a large number of 
	items with different topic mixtures or real-time responses are required.
Thus, our focus is on preprocessing individual topic influence such that when a diffusion item with 
	certain topic mixture comes, the online processing of finding the seed set is fast.
To do so, our first step is to collect two datasets in the past studies with available topic-aware
	influence analysis results on real networks and investigate their properties pertaining to our
	preprocessing purpose (Section~\ref{sec:obs}).
Our data analysis shows that in one network users and their relationships are largely separated by
	different topics while in the other network they have significant overlaps on different topics.
Even with this difference, a common property we find is that in both datasets most top
	seeds for a topic mixture come from top seeds of the constituent topics, which matches
	our intuition that influential individuals for a mixed item are usually influential in at least
	one topic category.

Motivated by our findings from the data analysis, we explore two preprocessing based algorithms
	(Section~\ref{sec:algo}).
The first algorithm, {\em Best Topic Selection} (BTS), minimizes online processing by simply using a seed set for
	one of the constituent topics.
Even for such a simple algorithm, we are able to provide a theoretical approximation ratio
	(when a certain property holds), and thus BTS serves as a baseline for preprocessing
	algorithms.
The second algorithm, {\em Marginal Influence Sort} (MIS), further uses pre-computed marginal influence
	of seeds on each topic to avoid slow greedy computation.
We provide a theoretical justification showing that MIS can be as good as the offline greedy
	algorithm when nodes are fully separated by topics.

We then conduct experimental evaluations of these algorithms and comparing them with both the greedy
	algorithm and a state-of-the-art heuristic algorithm PMIA~\cite{wang2012scalable}, on 
	the two datasets used in data analysis as well as a third dataset for testing scalability
	(Section~\ref{sec:experiments}).
From our results, we see that MIS algorithm stands out as the best candidate for preprocessing based
	real-time influence maximization: it finishes online processing within a few microseconds and its influence spread either
	matches or is very close to that of the greedy algorithm.
	
Our work, together with a recent independent work~\cite{aslayonline}, is one of the first that
	study topic-aware influence maximization with focus on preprocessing.
Comparing to~\cite{aslayonline}, our contributions include:
	(a) we include data analysis on two real-world datasets with learned influence parameters, which
		shows different topical influence properties and motivates our algorithm design;
	(b) we provide theoretical justifications to our algorithms; 
	(c) the use of marginal influence of seeds in individual topics in MIS is novel, and is complementary
	to the approach in~\cite{aslayonline}; 
	(d) even though MIS is quite simple, it achieves competitive influence spread within microseconds of online processing time
	rather than milliseconds needed in~\cite{aslayonline}.



\section{Preliminaries} \label{sec:pre}
In this section, we introduce the background and problem definition on the 
	topic-aware influence diffusion models.
We focus on the independent cascade model \cite{kempe2003maximizing} for ease of presentation, 
	but our results also hold for
	other models parameterized with edge parameters such as the linear threshold model \cite{kempe2003maximizing}.

\subsection{Independent cascade model}
We consider a social network as a directed graph $G=(V,E)$, where each node in $V$ represents a user,
and each edge in $E$ represents the relationship between two users.
For every edge $(u, v) \in E$, denote its {\em influence probability} as $p(u, v) \in [0, 1]$,
and for all $(u, v) \notin E$ or $u = v$, we assume $p(u, v) = 0$.

The {\em independent cascade (IC)} model, defined in \cite{kempe2003maximizing},
 	captures the stochastic process of contagion in discrete time.
Initially at time step $t=0$, a set of nodes $S \subseteq V$ called {\em seed nodes} are activated.
At any time $t \ge 1$, if node $u$ is activated at time $t-1$,
 	it has one chance of activating each of its inactive outgoing neighbor $v$
 	with probability $p(u,v)$.
A node stays active after it is activated.
This process stops when no more nodes are activated.
We define {\em influence spread} of seed set $S$ under influence probability function $p$, 
	denoted $\sigma(S, p)$, 
	as the expected number of active nodes after the diffusion process ends.
As shown in \cite{kempe2003maximizing}, for any fixed $p$, $\sigma(S, p)$
is monotone (i.e., $\sigma(S, p) \leq \sigma(T, p)$ for any $S \subseteq T$) and submodular
(i.e., $\sigma(S \cup \{ v \}, p) - \sigma(S, p) \geq
\sigma(T \cup \{ v \}, p) - \sigma(T, p)$ for any $S \subseteq T$ and $v \in V$)
on its seed set parameter.
The next lemma further shows that for any fixed $S$, $\sigma(S, p)$ is monotone in $p$.
For two influence probability functions $p$ and $p'$ on graph $G=(V,E)$, we denote
	$p \le p'$ if for any $(u,v)\in E$, $p(u,v) \le p'(u,v)$.
We say that influence spread function $\sigma(S, p)$ is {\em monotone in $p$} if
	for any $p\le p'$, we have $\sigma(S,p) \le \sigma(S,p')$.

\begin{lemma} \label{lemma:monotone}
For any fixed seed set $S \subseteq V$, $\sigma(S, p)$ is monotone in $p$.
\end{lemma}

\begin{proof}[Proof sketch]
We use the following coupling method.
For any edge $(u,v)\in E$, we select a number $x(u,v)$ uniformly at random in $[0,1]$.
Then for any influence probability function $p$, we select edge $(u,v)$ as a {\em live edge} if
	$x(u,v) \le p(u,v)$ and otherwise it is a {\em blocked edge}.
All live edges form a random {\em live-edge graph} $G_L(p)$.
One can verify that $\sigma(S,p)$ is the expected value of the size of node set 
	reachable from $S$ in random graph $G_L(p)$.
Moreover, for $p$ and $p'$ such that $p \le p'$, one can verify that
	after fixing the random numbers $x(u,v)'s$, live-edge
	graph $G_L(p)$ is a subgraph of live-edge graph $G_L(p')$, and thus nodes reachable from $S$
	in $G_L(p)$ must be also reachable from $S$ in $G_L(p')$.
This implies that $\sigma(S,p) \le \sigma(S,p')$. \QED
\end{proof}

We remark that using a similar idea as above we
	could show that influence spread in the linear threshold (LT) model \cite{kempe2003maximizing}
	is also monotone in the edge weight parameter.

\subsection{Influence maximization}


Given a graph $G=(V,E)$, an influence probability function $p$, and a budget $k$, 
{\em influence maximization} is the task of selecting at most $k$ seed nodes such that the influence spread
	is maximized, i.e., finding set $\Seed=\Seed(k,p)$ such that
$$
\Seed(k,p) = \argmax_{S \subseteq V, |S| \leq k} \sigma(S, p).
$$




In \cite{kempe2003maximizing}, Kempe et al. show that the influence maximization problem is NP-hard
	in both the IC model and the LT model.
They propose the greedy approach for
	influence maximization, as shown in Algorithm~\ref{alg:greedy}.
Given influence probability function $p$, the {\em marginal influence (MI)} of a node $v$ under
	seed set $S$ is defined as
	$\MI(v | S, p) = \sigma(S \cup \{v\}, p)-\sigma(S, p)$, for any $v\in V$.
The greedy algorithm selects $k$ seeds in $k$ iterations, and in the $j$-th iteration 
	it selects a node $v_j$
	with the largest marginal influence under the current seed set $S_{j-1}$ and adds $v_j$ into $S_{j-1}$
	to obtain $S_j$.
Kempe et al. use Monte Carlo simulations to obtain accurate estimates on marginal influence $\MI(v | S, p)$, 
	and later Chen et al. show that indeed exact computation of influence spread $\sigma(S,p)$ or
	marginal influence $\MI(v | S, p)$ is \#P-hard \cite{wang2012scalable}.
The monotonicity and submodularity of $\sigma(S,p)$ in $S$ guarantees that 
	the greedy algorithm selects a seed set with approximation ratio $1-\frac{1}{e} - \varepsilon$, that is,
	it returns a seed set $\Seedg = \Seedg(k, p)$ such that
$$
\sigma(\Seedg, p) \ge \left(1-\frac{1}{e} - \varepsilon \right) \sigma(\Seed,p),
$$
for any small $\varepsilon > 0$, where $\varepsilon$ accommodates the inaccuracy in Monte Carlo
	estimations.

\begin{algorithm}[t]
\begin{algorithmic}[1]
 \REQUIRE $G=(V,E)$, $p$, $k$.
 \STATE $S_0 = \emptyset$
    \FOR{$j = 1,2,\cdots,k$}
        \STATE $v_j = \argmax_{v \in V\setminus S_{j-1}} \MI(v | S_{j-1}, p)$ 
       \STATE $S_j = S_{j-1} \cup \{v_j\}$
     \ENDFOR
 \RETURN $S_k$
\end{algorithmic}
\caption{Greedy algorithm.}
\label{alg:greedy}
\end{algorithm}

\vspace{-3mm}
\subsection{Topic-aware independent cascade model and topic-aware influence maximization}

Topic-aware independent cascade (TIC) model \cite{barbieri2013topic} is an extension
of the IC model to incorporate topic mixtures in any diffusion item.
Suppose there are $d$ base topics, and we use set notation $[d] = \{1,2,\cdots,d\}$ 
to denote topic $1,2, \cdots, d$.
We regard each diffusion item as a distribution of these topics.
Thus, any item can be expressed as a vector $I=(\lambda_1, \lambda_2, \dots, \lambda_{d})$
where $\forall i \in [d]$, $\lambda_i \in [0,1]$ and $\sum_{i \in [d]} \lambda_i=1$. 
We also refer $(\lambda_1, \lambda_2, \dots, \lambda_{d})$ as a {\em topic mixture}.
Given a directed social graph $G=(V,E)$, for any topic $i \in [d]$, 
influence probability on that topic is $p_i: V \times V \rightarrow [0,1]$, 
and for all $(u, v) \notin E$ or $u = v$, we assume $p_i(u, v) = 0$.
In the TIC model, the influence probability function $p$ for any diffusion item $I=(\lambda_1,\lambda_2,\dots,\lambda_{d})$
is defined as $p(u,v) = \sum_{i \in [d]} \lambda_i {p_i}(u,v)$, for all $ u,v \in V$ (or simply $p = \sum_{i \in [d]} \lambda_i {p_i}$). 
Then, the stochastic diffusion process and influence spread $\sigma(S, p)$ are exactly the same as
	defined in the IC 
	model by using the influence probability $p$ on edges.


Given a social graph $G$, base topics $[d]$, influence probability function $p_i$ for each base
	topic $i$, a budget $k$ and an item $I=(\lambda_1,\lambda_2,\dots,\lambda_{d})$, the {\em topic-aware
	influence maximization} is the task of finding optimal seeds $\Seed = \Seed(k, p) \subseteq V$,
	where $p = \sum_{i \in [d]} \lambda_i {p_i}$, 
	to maximize the influence spread, i.e., $\Seed = \argmax_{S \subseteq V, |S| \leq k} \sigma(S, p)$.


\section{Data Observation} \label{sec:obs}

There are relatively few studies on topic-aware influence analysis. 
For our study, we are able to obtain datasets from two prior studies,
	one is on social movie rating network Flixster \cite{barbieri2013topic} and the other is
	on academic collaboration network Arnetminer \cite{tang2009social}.
In this section, we describe these two datasets, and present
	statistical observations on these datasets, which will help us in our algorithm design.

\subsection{Data description}
We obtain two real-world datasets, Flixster and Arnetminer, which include influence analysis results
	from their respective raw data, from the authors of the prior studies \cite{barbieri2013topic,tang2009social}.

Flixster\footnote{www.flixster.com} is an American social movie site for discovering new movies, learning about movies, and meeting others with similar tastes in movies. 
The raw data in Flixster dataset is the action traces of movie ratings of users.
The Flixster network represents users as nodes, and two users $u$ and $v$ are connected by a directed edge
	$(u,v)$ if they are friends both rating the same movie and $v$ rates the movie shortly later after $u$ does so.
The network contains 29357 nodes, 425228 directed edges and 10 topics 
	\cite{barbieri2013topic}. 
Barbieri et al. \cite{barbieri2013topic} use their proposed TIC model and apply maximum likelihood
	estimation method on the action traces to obtain influence probabilities on edges for all 10 topics.
We found that there are a disproportionate number of edges with influence probabilities higher than
	$0.99$, which is due to the lack of sufficient samplings of propagation events over these edges.
We smoothen these influence probability values by changing all the probabilities larger than $0.99$ 
	to random numbers according to the probability distribution of all the probabilities smaller than 
	$0.99$.
We also obtain 11659 topic mixtures, and demonstrate the distribution of the number of
	topics in item mixtures in Table~\ref{tab:dist_mixture}.
We eliminate individual probabilities that are too weak ($\forall i\in[d], \lambda_i < 0.01$).
In general, most items are on a single topic only, with some two-topic mixtures. Mixtures with three or four topics are already rare and there are no items with five or more topics.

%
\begin{table}[htbp]
  \centering
  \caption{Distribution of topic numbers of mixture items in Flixster}
  {\scriptsize
    \begin{tabular}{rccccc}
    \toprule
    \# Mixed topics & 1     & 2     & 3     & 4     & 5 \\
    \midrule
    \multicolumn{1}{l}{\# Samples} & 11285 & 354   & 18    & 2     & 0 \\
    \multicolumn{1}{l}{(Percentage)} & (96.79\%) & (3.04\%) & (0.15\%) & (0.02\%) & (0.00\%) \\
    \bottomrule
    \end{tabular}%
  }
  \label{tab:dist_mixture}%
\end{table}%

Arnetminer\footnote{arnetminer.org} is a free online service used to index and search academic social networks. 
The Arnetminer network represents authors as nodes and two authors have an edge if they coauthored 
	a paper.
The raw data in the Arnetminer dataset is not the action traces but the topic distributions of all nodes
	and the network structure \cite{tang2009social}.
Tang et al. apply factor graph analysis to obtain influence probabilities on edges
	from node topic distributions and the network structure \cite{tang2009social}.
The resulting network contains 5114 nodes, 34334 directed edges and 8 topics, and all 8 topics
	are related to computer science, such as data mining, machine learning, information retrieval, etc.
Mixed items propagated in such academic networks could be ideas or papers from related topic 
	mixtures, although there are no raw data of topic mixtures available in Arnetminer.

Tables~\ref{tab:stat1} and~\ref{tab:stat2} provide statistics for the learned influence probabilities
	for every topic in Arnetminer and Flixster dataset.
Column ``nonzero'' provides the number of edges having nonzero probabilities on the specific topic.
Other columns are mean, standard deviation, 25-percentile, 50-percentile (median), and 75-percentile
	of the probabilities among the nonzero entries.
The basic statistics show similar behavior between the two datasets, such as mean probabilities are
	mostly between $0.1$ and $0.2$, standard deviations are mostly between $0.1$ to $0.3$, etc.
Comparing among different topics, even though the means and other statistics are similar to one another,
	the number of nonzero edges may have up to 10 fold difference.
This indicates that some topics are more likely to propagate than others.


\begin{table}[t]
  \centering
  \caption{Influence probability statistics of Arnetminer} \label{tab:stat1}
  {\scriptsize
    \begin{tabular}{ccccccc}
    \addlinespace
    \toprule
    Topic & nonzero & mean & deviation & 25\%  & 50\%  & 75\% \\
    \midrule
    1     & 3355  & 0.175 & 0.230 & 0.023 & 0.075 & 0.229 \\
    2     & 13331 & 0.093 & 0.154 & 0.010 & 0.031 & 0.100 \\
    3     & 3821  & 0.158 & 0.214 & 0.020 & 0.065 & 0.201 \\
    4     & 1537  & 0.217 & 0.243 & 0.038 & 0.120 & 0.316 \\
    5     & 2468  & 0.197 & 0.262 & 0.018 & 0.080 & 0.266 \\
    6     & 1236  & 0.240 & 0.273 & 0.034 & 0.122 & 0.353 \\
    7     & 4439  & 0.145 & 0.222 & 0.011 & 0.046 & 0.177 \\
    8     & 3439  & 0.162 & 0.220 & 0.022 & 0.069 & 0.201 \\
    \bottomrule
    \end{tabular}%
   }
  \label{tab:stat_arnet}%
\end{table}%

\begin{table}[t]
  \centering
  \caption{Influence probability statistics  of Flixster} \label{tab:stat2}
  {\scriptsize
    \begin{tabular}{ccccccc}
    \addlinespace
    \toprule
    Topic & nonzero & mean & deviation & 25\%  & 50\%  & 75\% \\
    \midrule
    1     & 54032 & 0.173 & 0.215 & 1.00E-04 & 0.086 & 0.264 \\
    2     & 84322 & 0.172 & 0.227 & 4.36E-05 & 0.067 & 0.260 \\
    3     & 231807 & 0.089 & 0.146 & 1.18E-04 & 0.024 & 0.112 \\
    4     & 35394 & 0.162 & 0.226 & 6.78E-03 & 0.050 & 0.250 \\
    5     & 118125 & 0.097 & 0.141 & 2.45E-03 & 0.037 & 0.131 \\
    6     & 37489 & 0.090 & 0.142 & 6.85E-03 & 0.033 & 0.100 \\
    7     & 84716 & 0.166 & 0.230 & 3.12E-05 & 0.050 & 0.250 \\
    8     & 149140 & 0.097 & 0.145 & 9.01E-04 & 0.036 & 0.131 \\
    9     & 152181 & 0.103 & 0.158 & 2.14E-04 & 0.032 & 0.140 \\
    10    & 139335 & 0.159 & 0.235 & 3.27E-05 & 0.029 & 0.250 \\
    \bottomrule
    \end{tabular}%
  }
  \label{tab:stat_flix}%
\end{table}%

\subsection{Topic separation on edges and nodes}
For the two datasets, we would like to investigate how different topics overlap on edges and nodes.
To do so, we define the following
coefficients to characterize the properties of a social graph.



Given threshold $\theta\geq 0$, for every topic $i$, denote edge set $\efilter_i(\theta) =  
	\{ (u,v) \in E\,|\, {p_i}(u,v) > \theta \}$, and node set $\vfilter_i(\theta) = \{v\in V \,|\, \sum_{u:(v,u)\in E} p_i(v,u)+ \sum_{u:(u,v)\in E} p_i(u,v) > \theta\}$.
For topics $i$ and $j$, we define {\em edge overlap coefficient} as
	$\reo_{ij}(\theta) =
	\frac{|\efilter_i(\theta) \cap \efilter_j(\theta)|}
	{\min\{|\efilter_i(\theta)|, |\efilter_j(\theta)|\}}$, and
	{\em node overlap coefficient} as $ \rvo_{ij}(\theta) =
	\frac{|\vfilter_i(\theta) \cap \vfilter_j(\theta)|}
		{\min\{|\vfilter_i(\theta)|, |\vfilter_j(\theta)|\}}$.
If $\theta$ is small and the overlap coefficient is small, it means that the two topics are fairly
	separated in the network.
In particular, we say that the network is {\em fully separable} for topics $i$ and $j$
	if $\rvo_{ij}(0) = 0$,
	and it is fully separable for all topics if 
	$\rvo_{ij}(0) = 0$ for any pair $i$ and $j$ with $i\ne j$.
Then we apply the above coefficients to the Flixster and Arnetminer datasets. 

Table~\ref{tab:sep_armetminer} shows the edge and node overlap coefficients with
	threshold $\theta=0.1$ for every pair of topics in the Arnetminer dataset.
Correlating with Table~\ref{tab:stat2}, we see that $\theta=0.1$ is around the mean value for all topics.
Thus it is a reasonably small value especially for the node overlap coefficients, which is about 
	aggregated probability of all edges incident to a node.
A clear indication in Table~\ref{tab:sep_armetminer} is that topic overlap on both edges and nodes are very small
	in Arnetminer, with most node overlap coefficients less than $5\%$.
We believe that this is because in academic collaboration network, most researchers work on one specific
	research area, and only a small number of researchers work across different research areas.

Tables~\ref{tab:edge_flixster} and~\ref{tab:node_flixster3} show the edge and node overlap coefficients
	for the Flixster dataset.
Different from the Arnetminer dataset, both edges and nodes have significant overlaps.
For edge overlaps, even with threshold $\theta=0.3$, all topic pairs have edge overlap between $15\%$ and $40\%$.
For node overlap, we test the threshold for both $0.5$ to $5$, but the overlap coefficients do not
	significantly change: at $\theta=5$, most pairs still have above $60\%$ and up to $89\%$ 
	overlap.
We think that this could be explained by the nature of Flixster, which is a movie rating site.
Most users are interested in multiple categories of movies, and their influence to their friends
	are also likely to be across multiple categories.
It is interesting to see that, even though the per-topic statistics between Arnetminer and Flixster 
	are similar, they show quite different cross-topic overlap behaviors, which can be explained by
	the nature of the networks.
This could be an independent research topic for further investigations on the influence behaviors among
	different topics.

\begin{table}[t]
  \centering
  \caption{Edge and Node overlap coefficients on Arnetminer. The upper black triangle represents edge overlap coefficient when $\theta=0.1$. The entry on row $i$, column $j$ represents $\reo_{ij}(0.1)$; the lower blue triangle represents node overlap coefficient when $\theta=0.1$. The entry on row $i$, column $j$ represents $\rvo_{ij}(0.1)$.}
  {\scriptsize
    \begin{tabular}{cccccccc}
		- & 0.017 & 0.002 & 0.000 & 0.005 & 0.006 & 0.000 & 0.022 \\
	    \color{blue}0.068 & - & 0.001 & 0.004 & 0.001 & 0.001 & 0.002 & 0.000 \\
	    \color{blue}0.018 & \color{blue}0.014 & - & 0.000 & 0.000 & 0.001 & 0.000 & 0.000 \\
	    \color{blue}0.002 & \color{blue}0.029 & \color{blue}0.000 & - & 0.000 & 0.011 & 0.017 & 0.000 \\
	    \color{blue}0.025 & \color{blue}0.005 & \color{blue}0.005 & \color{blue}0.000 & - & 0.000 & 0.000 & 0.015 \\
	    \color{blue}0.054 & \color{blue}0.049 & \color{blue}0.049 & \color{blue}0.011 & \color{blue}0.000 & - & 0.009 & 0.001 \\
	    \color{blue}0.006 & \color{blue}0.025 & \color{blue}0.003 & \color{blue}0.017 & \color{blue}0.007 & \color{blue}0.063 & - & 0.000 \\
	    \color{blue}0.108 & \color{blue}0.001 & \color{blue}0.008 & \color{blue}0.000 & \color{blue}0.079 & \color{blue}0.011 & \color{blue}0.004 & - \\

    \end{tabular}%
   }
  \label{tab:sep_armetminer}%
\end{table}%

\begin{table}[t]
  \centering
  \caption{Edge overlap coefficients on Flixster. The upper black triangle represents edge overlap coefficient when $\theta=0.1$. The entry on row $i$, column $j$ represents $\reo_{ij}(0.1)$; the lower blue triangle represents edge overlap coefficient when $\theta=0.3$. The entry on row $i$, column $j$ represents $\reo_{ij}(0.3)$.}
   {\scriptsize
    \begin{tabular}{cccccccccc}
 	  -  & 0.33  & 0.49  & 0.27  & 0.36  & 0.35  & 0.35  & 0.42  & 0.43  & 0.39 \\
      \color{blue}0.22  & -  & 0.48  & 0.33  & 0.31  & 0.41  & 0.31  & 0.36  & 0.38  & 0.39 \\
      \color{blue}0.28  & \color{blue}0.26  & -  & 0.46  & 0.50  & 0.48  & 0.55  & 0.50  & 0.57  & 0.52 \\
     \color{blue} 0.15  & \color{blue}0.19  & \color{blue}0.22  & -  & 0.33  & 0.25  & 0.31  & 0.37  & 0.38  & 0.38 \\
      \color{blue}0.20  & \color{blue}0.25  & \color{blue}0.34  & \color{blue}0.13  & -  & 0.52  & 0.30  & 0.46  & 0.45  & 0.37 \\
      \color{blue}0.23  & \color{blue}0.29  & \color{blue}0.28  & \color{blue}0.16  & \color{blue}0.31  & -  & 0.36  & 0.50  & 0.47  & 0.38 \\
      \color{blue}0.25  & \color{blue}0.21  & \color{blue}0.34  & \color{blue}0.18  & \color{blue}0.24  & \color{blue}0.25  & -  & 0.37  & 0.43  & 0.46 \\
      \color{blue}0.21  & \color{blue}0.24  & \color{blue}0.38  & \color{blue}0.15  & \color{blue}0.31  & \color{blue}0.29  & \color{blue}0.25  & -  & 0.44  & 0.37 \\
      \color{blue}0.24  & \color{blue}0.24  & \color{blue}0.44  & \color{blue}0.17  & \color{blue}0.32  & \color{blue}0.28  & \color{blue}0.29  & \color{blue}0.35  & -  & 0.42 \\
      \color{blue}0.28  & \color{blue}0.27  & \color{blue}0.47  & \color{blue}0.23  & \color{blue}0.29  & \color{blue}0.26  & \color{blue}0.35  & \color{blue}0.32  & \color{blue}0.37  & - \\

	\end{tabular}%
  }
  \label{tab:edge_flixster}%
\end{table}%

\begin{table}[t]
  \centering
  \caption{Node overlap coefficients on Flixster. The upper black triangle represents node overlap coefficient when $\theta=0.5$. The entry on row $i$, column $j$ represents $\rvo_{ij}(0.5)$; the lower blue triangle represents node overlap coefficient when $\theta=5.0$. The entry on row $i$, column $j$ represents $\rvo_{ij}(5.0)$.}
    {\scriptsize
    \begin{tabular}{cccccccccc}
	-  & 0.79  & 0.91  & 0.68  & 0.76  & 0.81  & 0.77  & 0.83  & 0.85  & 0.87 \\
    \color{blue}0.69  & -  & 0.88  & 0.82  & 0.76  & 0.88  & 0.75  & 0.74  & 0.77  & 0.84 \\
    \color{blue}0.83  & \color{blue}0.64  & -  & 0.93  & 0.92  & 0.95  & 0.91  & 0.92  & 0.91  & 0.87 \\
    \color{blue}0.53  & \color{blue}0.67  & \color{blue}0.75  & -  & 0.77  & 0.63  & 0.78  & 0.83  & 0.85  & 0.89 \\
    \color{blue}0.58  & \color{blue}0.70  & \color{blue}0.87  & \color{blue}0.50  & -  & 0.90  & 0.73  & 0.84  & 0.85  & 0.85 \\
    \color{blue}0.76  & \color{blue}0.83  & \color{blue}0.86  & \color{blue}0.46  & \color{blue}0.91  & -  & 0.86  & 0.93  & 0.92  & 0.91 \\
    \color{blue}0.71  & \color{blue}0.53  & \color{blue}0.72  & \color{blue}0.62  & \color{blue}0.72  & \color{blue}0.78  & -  & 0.77  & 0.81  & 0.88 \\
    \color{blue}0.72  & \color{blue}0.57  & \color{blue}0.82  & \color{blue}0.60  & \color{blue}0.85  & \color{blue}0.89  & \color{blue}0.59  & -  & 0.83  & 0.84 \\
    \color{blue}0.74  & \color{blue}0.53  & \color{blue}0.84  & \color{blue}0.62  & \color{blue}0.82  & \color{blue}0.89  & \color{blue}0.63  & \color{blue}0.73  & -  & 0.83 \\
    \color{blue}0.89  & \color{blue}0.74  & \color{blue}0.81  & \color{blue}0.83  & \color{blue}0.88  & \color{blue}0.89  & \color{blue}0.82  & \color{blue}0.82  & \color{blue}0.84  & - \\
    \end{tabular}%
    }
  \label{tab:node_flixster3}%
\end{table}

\begin{table}[t]
\setlength{\tabcolsep}{10pt}
\centering
\caption{Overlap coefficient statistics for all topic pairs} \label{tab:overlap}
\begin{tabular}{cccc}
	\toprule
			& min	& mean	& max \\
	\midrule
	Arnetminer: {\small  $\reo_{ij}(0.1)$}  & 0  & 0.0041 & 0.022  \\
	Arnetminer: {\small $\rvo_{ij}(0.1)$}  & 0  & 0.0236 & 0.108 \\
	Flixster: {\small $\reo_{ij}(0.1)$} & 0.25 & 0.4058 & 0.57 \\
	Flixster: {\small $\reo_{ij}(0.3)$} & 0.13 & 0.2662 & 0.47 \\
	Flixster: {\small $\rvo_{ij}(0.5)$} & 0.63 & 0.836 & 0.95 \\
	Flixster: {\small $\rvo_{ij}(5.0)$} & 0.46 & 0.734 & 0.91 \\
	\bottomrule
\end{tabular}
\end{table}

Table~\ref{tab:overlap} summarizes the edge and node overlap coefficient statistics among all pairs of
	topics for the two datasets.
We can see that Arnetminer network has fairly separate topics on both nodes and edges, while
	Flixter network have significant topic overlaps. 
This may be explained by that in an academic network most researchers only work in one research area,
	but in a movie network many users are interested in more than one type of movies. Therefore, our first observation is: 
\begin{observation}
Topic separation in terms of influence probabilities is network dependent.
In the Arnetminer network, topics are mostly separated among different edges and nodes in the network, while
	in the Flixster network there are significant overlaps on topics among nodes and edges.
\label{obv:topic}
\end{observation}

\subsection{Sources of seeds in the mixture}
Our second observation is more directly related to influence maximization.
We would like to see if seeds selected by the greedy algorithm for a topic mixture are likely coming
	from top seeds for each individual topic.
Intuitively, it seems reasonable to assume that top influencers for a topic mixture are likely from 
	top influencers in their constituent topics.

\begin{table}[t]
  \centering
\caption{Percentage of seeds in topic mixture that are also seeds of constituent topics.}
{\scriptsize
\begin{tabular}{cccc}
	\toprule
		& Arnetminer	& Flixster (random)	& Flixster (Dirichlet)\\
	\midrule
	Seeds overlap	& 94.80\%	& 81.16\% &	85.24\% \\
	\bottomrule
\end{tabular}
}
\label{tab:percentage}%
\end{table}%

To check the source of seeds, we randomly generate 50 mixtures of two topics for both Arnetminer and
	Flixster, and use the greedy algorithm to select seeds for the mixture and the constituent topics.
We then check the percentage of seeds in the mixture that is also in the constituent topics.
Table~\ref{tab:percentage}
shows our test results (Flixster (Dirhilect) is the result using a Dirichlet
distribution to generate topic mixtures, see Section~\ref{sec:experiments}
for more details).
Our observation below matches our intuition:

\begin{observation}
Most seeds for topic mixtures come from the seeds of constituent topics, in both Arnetminer and
	Flixster networks.
\label{obv:seed}
\end{observation}

For Arnetminer, it is likely due to the topic separation as observed in Table~\ref{tab:sep_armetminer}.
For Flixster, even though topics have significant overlaps, 
	these overlaps may result in many shared seeds between topics, which would also contribute
	as top seeds for topic mixtures.



\section{Preprocessing Based Algorithms} \label{sec:algo}
Topic-aware influence maximization can be solved by using existing influence maximization algorithms
	such as the ones in \cite{kempe2003maximizing,wang2012scalable}: when a query on 
	an item $I = (\lambda_1, \lambda_2, \cdots, \lambda_d)$ comes, the algorithm first computes
	the mixed influence probability function $p = \sum_j \lambda_j p_j$, and then applies existing
	algorithms using parameter $p$.
This, however, means that for each topic mixture influence maximization has to be carried out from scratch,
	which could be inefficient in large-scale networks.

In this section, motivated by observations made in Section~\ref{sec:obs},
	we introduce two preprocessing based algorithms that cover different design choices.
The first algorithm Best Topic Selection focuses on minimizing online processing time,
	and the second one MIS uses pre-computed marginal influence to achieve both fast online processing
	and competitive influence spread.
For convenience, we consider the budget $k$ as fixed in our algorithms, but we could extend the 
	algorithms  
	to consider multiple $k$ values in preprocessing.

\subsection{Best Topic Selection (BTS) algorithm}
The idea of our first algorithm is to minimize online processing by simply selecting a seed set
	for one of the constituent topics in the topic mixture that has the best influence performance,
	and thus we call it Best Topic Selection (BTS) algorithm.
More specifically, given an item $I = (\lambda_1, \lambda_2, \cdots, \lambda_{d})$, 
	if we have pre-computed the seed set $\Seedg_i = \Seedg(k, \lambda p_i)$ via the greedy algorithm for each
	topic $i$, then we would simply use the seed set $\Seedg_{i'}$ that gives the best influence spread,
	i.e., $i' = \argmax_{i \in [d]} \sigma(\Seedg_i, \lambda_i p_i)$.
However, in the preprocessing stage, the topic mixture 
	$(\lambda_1, \lambda_2, \cdots, \lambda_{d})$ is not guaranteed to be pre-computed exactly.
To deal with this issue, we pre-compute influence spread for a number of landmark points for each topic,
	and use rounding method in online processing to complete seed selection, as we explain
	in more detail now.

Denote constant set $\Lambda = \{\clambda_0, \clambda_1, \clambda_2, \cdots, \clambda_m\}$ as a
	set of {\em landmarks}, 
where $0 = \clambda_0 < \clambda_1 < \cdots < \clambda_m = 1$.
For each $\lambda \in \Lambda$ and  each topic $i \in [d]$,
	we pre-compute $\Seedg(k, \lambda p_i)$ and $\sigma(\Seedg(k, \lambda p_i), \lambda p_i)$ in the preprocessing stage, and store these values for online processing.
In our experiments, we use uniformly selected landmarks and show that they are good enough for influence
	maximization.
More sophisticated landmark selection method may be applied, 
	such as the machine learning based method in~\cite{aslayonline}.

We define two rounding notations that return one of the neighboring landmarks in $\Lambda = \{\clambda_0, \clambda_1, \cdots, \clambda_m\}$:
for any $\lambda \in [0,1]$, $\ulambda$ is denoted as
rounding $\lambda$ down to $\clambda_j$ where $\clambda_j \leq \lambda < \clambda_{j+1}$
and $\clambda_j, \clambda_{j+1} \in \Lambda$,
and  $\olambda$ as rounding up to $\clambda_{j+1}$ where
$\clambda_j < \lambda \leq \clambda_{j+1}$ and $\clambda_j, \clambda_{j+1} \in \Lambda$.

Given 
$I = (\lambda_1, \lambda_2, \cdots, \lambda_{d})$,
let $D^+_I = \{i\in [d] \,|\, \lambda_i > 0 \}$.
With the pre-computed $\Seedg(k, \lambda p_i)$ and $\sigma(\Seedg(k, \lambda p_i), \lambda p_i)$
	for every $\lambda \in \Lambda$ and every topic $i$, the BTS algorithm is given in Algorithm~\ref{alg:top}.
The algorithm basically rounds down
	the mixing coefficient on every topic to $(\ulambda_1, \cdots, \ulambda_d)$, and then
	returns the seed set $\Seedg(k, \ulambda_{i'} p_{i'})$ that gives the largest influence spread at the round-down
	landmarks: $i' = \argmax_{i \in D^+_I} \sigma(\Seedg(k, \ulambda_i p_{i}), \ulambda_i p_i)$.

\begin{algorithm}[t]
\begin{algorithmic}[1]
	\REQUIRE $G=(V,E)$, $k$, $\{p_i\,|\, i\in [d]\}$, $I = (\lambda_1, \cdots, \lambda_{d})$, $\Lambda$,
				$\Seedg(k, \lambda p_i)$ and $\sigma(\Seedg(k, \lambda p_i), \lambda p_i)$, $\forall \lambda \in \Lambda, \forall i\in [d]$.
	\STATE $I' = (\ulambda_1, \cdots, \ulambda_d)$
	\STATE $i' = \argmax_{i \in D^+_I} \sigma(\Seedg(k, \ulambda_i p_i), \ulambda_i p_i)$ \label{line:topselect}
	\RETURN $\Seedg(k, \ulambda_{i'} p_{i'})$
\end{algorithmic}
	\caption{Best Topic Selection (BTS) Algorithm}
	\label{alg:top}
\end{algorithm}


BTS is rather simple since it directly outputs a seed set for one of the constituent topics.
However, we show below that even such a simple scheme could provide a theoretical approximation guarantee
	(if the influence spread function is sub-additive as defined below).
Thus, we use BTS as a baseline for preprocessing based algorithms.

We say that influence spread function $\sigma(S,p)$ is {\em $c$-sub-additive in $p$} for some constant
	$c$ if for every set $S \subseteq V$ with $|S| \leq k$ and every mixture $(\lambda_1, \lambda_2,
		\ldots, \lambda_d)$, 
	$\sigma(S, \sum_{i \in D^+_I} \lambda_i p_i)$ $ \leq$ $ c \sum_{i \in D^+_I} \sigma(S, \lambda_i p_i)$.
The sub-additivity property above means that the influence spread of any seed set $S$ in any topic mixture
	will not exceed constant times of the sum of the influence spread of the same seed set for each individual topic.
It is easy to verify that, when the network is fully separable for all topic pairs, 
	$\sigma(S,p)$ is $1$-sub-additive.
The only counterexample to the sub-additivity assumption that we could find is 
	a tree structure where even layer edges are for one topic
	and odd layer edges are for another topic.
Such structures are rather artificial, and we believe that for real networks 
	the influence spread is $c$-sub-additive in $p$ with a reasonably small constant $c$.

%

We define $\mu_{\max} = \max_{i \in [d], \lambda \in [0,1]}
  \frac{\sigma(\Seedg(k, \olambda p_i), \olambda p_i)}
    {\sigma(\Seedg(k, \ulambda p_i), \ulambda p_i)}$, which is a value controlled by preprocessing.
A fine-grained landmark set $\Lambda$ could make $\mu_{\max}$ close to $1$. 
The following Theorem~\ref{thm:top} guarantees the theoretical approximation ratio of Algorithm~\ref{alg:top}.

\begin{theorem} \label{thm:top}
If the influence spread function $\sigma(S,p)$ is $c$-sub-additive in $p$,
Algorithm~\ref{alg:top} achieves $\frac{1-e^{-1}}{c |D^+_I| \mu_{\max}}$ approximation ratio for item $I = (\lambda_1, \lambda_2, \cdots, \lambda_{d})$.
\end{theorem}

\begin{proof}
Denote $\Seed = \Seed(k, p)$, $\oSeed_i = \Seed(k, \olambda_i p_i)$,
$\oSeedg_i = \Seedg(k, \olambda_i p_i)$ and 
$\uSeedg_i = \Seedg(k, \ulambda_i p_i)$. 
Since $\sigma(S, p)$ is monotone (Lemma~\ref{lemma:monotone}) and $c$-sub-additive in $p$, 
	it implies 
$\sigma(\Seed, p) = \sigma(\Seed, \sum_{i \in D^+_I} \lambda_i p_i)
	\leq c \sum_{i \in D^+_I} \sigma(\Seed, \lambda_i p_i)$
	$\leq$ $c \sum_{i \in D^+_I} \sigma(\Seed, \olambda_i p_i)$.
From \cite{kempe2003maximizing}, we know
$\sigma(\Seed(k, p_0), p_0) \leq \frac{1}{1-e^{-1}} \sigma(\Seedg(k, p_0), p_0)$
holds for any $p_0$ in Algorithm~\ref{alg:greedy}.
Thus we have, for each $i \in D^+_I$,
$
\sigma(\Seed, \olambda_i p_i)
	\leq \sigma(\oSeed_i, \olambda_i p_i)
	\leq \frac{\sigma(\oSeedg_i, \olambda_i p_i)}{1-e^{-1}} 
	\leq \frac{\mu_{\max} \cdot \sigma(\uSeedg_i, \ulambda_i p_i)}{1-e^{-1}}$.
According to line~\ref{line:topselect}
of Algorithm~\ref{alg:top},
$i'$ satisfies
$\sigma(\uSeedg_{i'}, \ulambda_{i'} p_{i'}) = \max_{i \in D^+_I} \sigma(\uSeedg_i, \ulambda_i p_i)$, 
and $\sigma(\uSeedg_{i'}, \ulambda_{i'} p_{i'}) \leq \sigma(\uSeedg_{i'}, \lambda_{i'} p_{i'})$.
Thus, connecting all the inequalities, we have
$\sigma(\Seed, p) 
	$ $\leq$ $ \frac{c |D^+_I| \mu_{\max}}{1-e^{-1}} \sigma(\uSeedg_{i'}, \lambda_{i'} p_{i'})$.
Therefore, Algorithm~\ref{alg:top}
achieves approximation ratio of $\frac{1}{c |D^+_I| \mu_{\max}}(1-\frac{1}{e})$ 
under the sub-additive assumption.  \QED
\end{proof}

The approximation ratio given in the theorem is a conservative bound for the worst case 
(e.g., a common setting may be $c=1$, $\mu_{\max}=1.5$, $|D^+_I|=2$).
Tighter online bound in our experiment section based on \cite{leskovec2007cost}
	shows that Algorithm~\ref{alg:top} performs much better than the worst case scenario.

\subsection{Marginal Influence Sort (MIS) algorithm}

Our second algorithm derives the seed set from pre-computed seed set of constituent topics, which is
	based on Observation~\ref{obv:seed}.
Moreover, it uses
	marginal influence information pre-computed to help select seeds from different seed sets.
Our idea is partially motivated from Observation~\ref{obv:topic}, especially the observation
	on Arnetminer dataset, which shows that in some cases the network could be well separated among
	different topics.
Intuitively, if nodes are separable among different topics, and each node $v$ is only pertinent to one topic
	$i$,
	the marginal influence of $v$ would not change much whether it is for a mixed item or the pure topic
		$i$.
The following lemma makes this intuition precise for the extreme case of fully separable networks.

\begin{lemma}\label{lemma:separable}
If a network is fully separable among all topics, then for any $v \in V$ and topic $i \in [d]$ 
	such that $\sigma(v, p_i) > 1$, 
	for any item $I = (\lambda_1, \lambda_2, \dots, \lambda_{d})$, for any seed set $S \subseteq V$,
	we have $\MI(v | S, \lambda_i p_i) = \MI(v | S, p)$, where $p = \sum_{j\in[d]} \lambda_j p_j$.
\end{lemma}

\begin{proof}[Proof sketch]
Let $G_i=(V_i,E_i)$ be the subgraph of $G$ generated by edges $(u,w)$ such that $p_i(u,w) > 0$ and their
	incident nodes.
It is easy to verify that when the network is fully separable among all topics, 
	$G_i$ and $G_j$ are disconnected for any $i\ne j$. 
In this case, we have
	(a) for any node $v$ and topic $i$ such that $\sigma(v, p_i) > 1$, $v \in V_i$; 
	(b) for any edge $(u,w)\in E_i$, $p(u,w) = \lambda_i p_i(u,w)$; and 
	(c) $\sigma(S,p') = \sum_{j\in [d]} \sigma(S\cap V_j, p')$ for any $p'$.
With the above property, a simple derivation following the definition of marginal influence 
	will lead to $\MI(v | S, \lambda_i p_i) = \MI(v | S, p)$. \QED
\end{proof}


%
%

The above lemma suggests that we can use the marginal influence of a node on each topic when dealing
	with a mixture of topics.
Algorithm MIS is based on this idea.

Recall the detail of Algorithm~\ref{alg:greedy}, given any fixed probability $p$ and budget $k$, for each iteration $j = 1,2,\cdots, k$,
it calculates $v_j$ to maximize marginal influence $\MI(v_j | S_{j-1}, p)$ and let $S_j = S_{j-1} \cup \{ v_j \}$
every time, and output $\Seedg(k, p) = S_k$ as seeds.
Let $\MII(v_j, p) = \MI(v_j | S_{j-1}, p)$, if $v_j \in \Seedg(k, p)$, and $0$ otherwise.
$\MII(v_j, p) $ is the marginal influence of $v_j$ according to the greedy selection order.

The preprocessing goes as follows. 
We also use the landmark set $\Lambda = \{\clambda_0, \clambda_1, \clambda_2, \cdots, \clambda_m\}$.
For every $\lambda \in \Lambda$, 
we pre-compute $\Seedg(k, \lambda p_i)$, for every single topic $i \in [d]$,
and cache $\MII(v, \lambda p_i)$, $\forall v \in \Seedg(k, \lambda p_i)$ in advance by Algorithm~\ref{alg:greedy}.

With the above preparation, we can design Marginal Influence Sort (MIS) algorithm as described in Algorithm~\ref{alg:MIS}.
Given an item $I = (\lambda_1, \cdots, \lambda_{d})$, the online processing stage first 
	rounding down the mixture to $I'= (\ulambda_1, \cdots, \ulambda_d)$, and then 
	use the union $V^g = \cup_{i \in [d], \ulambda_i > 0} \Seedg(k, \ulambda_i p_i)$ as seed candidates.
If a node appears in multiple pre-computed seed sets, we add their marginal influence 
	in each set together (line~\ref{line:addMI}).
Then we simply sort all nodes in $V^g$ according to their computed marginal influence $f(v)$ and
	return the top $k$ nodes as seeds.

\begin{algorithm}[t]
\begin{algorithmic}[1]
	\REQUIRE $G=(V,E)$, $k$, $\{p_i\,|\, i\in [d]\}$, $I = (\lambda_1, \cdots, \lambda_{d})$, $\Lambda$,
			$\Seedg(k, \lambda p_i)$ and $\MII(v, \lambda p_i)$, $\forall \lambda \in \Lambda$,
			$\forall i \in [d]$.
	\STATE $I' = (\ulambda_1, \cdots, \ulambda_d)$
	\STATE $V^g = \cup_{i \in [d], \ulambda_i > 0} \Seedg(k, \ulambda_i p_i)$
	\FOR{$v \in V^g$}
		\STATE $f(v) = \sum_{i \in [d], \ulambda_i > 0} \MII(v, \ulambda_i p_i)$ \label{line:addMI}
	\ENDFOR
	\RETURN top $k$ nodes with the largest $f(v), \forall v \in V^g$
\end{algorithmic}
	\caption{Marginal Influence Sort (MIS) Algorithm}
	\label{alg:MIS}
\end{algorithm}


Although MIS is a heuristic algorithm, it does guarantee the same performance as the
	original greedy algorithm in fully separable
	networks when the topic mixtures is from the landmark set, as shown by the theorem below.
Note that in a fully separable network, it is reasonable to assume that seeds for one topic comes from
	the subgraph for that topic, and thus seeds from different topics are disjoint.
\begin{theorem} \label{thm:MIS}

Suppose $I = (\lambda_1, \lambda_2, \cdots, \lambda_d)$, where each $\lambda_i \in \Lambda$, and
	$\Seedg(k, \lambda_1 p_1)$, 
	$\cdots$, $\Seedg(k, \lambda_d p_d)$
	are disjoint.
If the network is fully separable for all topics, 
the seed set calculated by Algorithm~\ref{alg:MIS}
is one of the possible sequences generated by Algorithm~\ref{alg:greedy}
under the mixed influence probability
$p = \sum_{i \in [d]} \lambda_i p_i$.
\end{theorem}

\begin{proof}[Proof sketch]
Denote $v_1, v_2, \cdots , v_k \in V^g$ as the final seeds selected for the topic mixture
	in this order, and let $S_0 = \emptyset$
and $S_\ell = S_{\ell - 1} \cup \{ v_\ell \}$, for $\ell = 1,2,\cdots,k$. 
Since the network is fully separable and topic-wise seed sets are disjoint, by Lemma~4.1 
	we can get that $v_1, v_2, \cdots , v_k$  are selected  from topic-wise seeds sets,
	and $\forall v \in V^g$, $f(v) = \MI(v | S_{\ell - 1}, p)$.
We can prove that $v_{\ell} = \argmax_{v \in V \setminus S_{\ell-1}}$ $
	\MI(v | S_{\ell-1}, p)$, $\forall \ell = 1,2,\cdots,k$ by induction.
It is straightforward to see that $v_1 = \argmax_{v \in V}$ $ \MI(v | \emptyset, p)$.
Assume it holds for $\ell = j \in \{1,2,\cdots,k-1\}$.
Then, for $\ell = j+1$, for a contradiction we suppose that
the $(j+1)$-th seed $v'$ is chosen from $V \setminus V^{g}$ other than $v_{j+1}$,
i.e., $\MI(v' | S_{j}, p) > \MI(v_{j+1} | S_{j}, p)$. Denote $i'$ such that $\sigma(v', p_{i'}) > 1$.
Since budget $k > j$, we can find 
	a node $u \in \Seedg(k, \lambda_{i'} p_{i'}) \setminus S_{j}$,
	such that $\MI(u | S_{j}, \lambda_{i'} p_{i'})$ $\geq$ $\MI(v' | S_{j}, \lambda_{i'} p_{i'})$,
and $u$ is selected before $v_{j+1}$, which is a contradiction.
Therefore, we will conclude that $v_1, v_2$, $\cdots$, $v_k$ is one possible sequence from 
	the greedy algorithm. \QED
\end{proof}

The theorem suggests that MIS would work well for networks that are fairly separated among different
	topics, which are verified by our test results on the Arnetminer dataset.
Moreover, even for networks that are not well separated, it is reasonable to assume that the marginal
	influence of nodes in the mixture is related to the sum of its marginal influence in
	individual topics, and thus we expect MIS to work also competitively in this case, which is
	verified by our test results on the Flixster dataset.

\section{Experiments}\label{sec:experiments}

We test the effectiveness of our algorithms by using a number of real-world datasets, and compare them with
	state-of-the-art influence maximization algorithms.

\subsection{Algorithms for comparison}

In our experiments, we test our topic-aware preprocessing based algorithms MIS and BTS comprehensively.
We also select three classes of algorithms for comparison:
(a) Topic-aware algorithms: The topic-aware greedy algorithm (\TAGreedy) and a state-of-the-art fast
		heuristic algorithm PMIA (\TAPMIA) from \cite{wang2012scalable}; 
(b) Topic-oblivious algorithms: The topic-oblivious greedy algorithm (\TOGreedy), degree algorithm (\TODegree) and random algorithm (\TORandom);
(c) Simple heuristic algorithms that do not need preprocessing: The topic-aware PageRank algorithm (\TAPageRank) from \cite{Brin98web} and WeightedDegree algorithm (\TAWeightDegree).

The greedy algorithm we use employs lazy evaluation \cite{leskovec2007cost} to provide hundreds of time of speedup to
	the original Monte Carlo based greedy algorithm \cite{kempe2003maximizing}, and also provides the best theoretical guarantee. 
PMIA is a fast heuristic algorithm for the IC model based on trimming influence propagation to a tree
	structure and fast recursive computation on trees, and it achieves thousand fold speedup comparing
	to optimized greedy approximation algorithms with a small degradation on influence spread
	\cite{wang2012scalable} (in this paper, we set a small threshold $\theta = 1/1280$
	to alleviate the degradation). 

Topic-oblivious algorithms work under previous IC model that does not identify topics, 
i.e., it takes the fixed mixture $\forall j\in [d], \lambda_j=\frac{1}{d}$.
{\TOGreedy} runs greedy algorithm for previous IC model and uses the top-$k$ nodes as its seeds.
{\TODegree} outputs the top-$k$ nodes with the largest degree based on the original graph. {\TORandom} 
simply chooses $k$ nodes at random.

We also carefully choose two simple heuristic algorithms
	that do not need preprocessing.
{\TAPageRank} uses the probability of the topic mixture as its transfer probability, and runs PageRank algorithm
	to select $k$ nodes with top rankings. The damping factor is set to $0.85$.
{\TAWeightDegree} uses the degrees weighted by the probability from topic mixtures, and selects top-$k$ nodes 	with the highest weighted degrees.

Finally, we study the possibility of acceleration for large graphs by comparing PMIA
with greedy algorithm in preprocessing stage.
Therefore, we denote MIS and BTS algorithms, utilizing the seeds and marginal influence 
from greedy and PMIA, as {\MISGreedy}, {\BTSGreedy} and {\MISPMIA}, {\BTSPMIA}, respectively.

\subsection{Experiment setup}

We conduct all the experiments on a computer with 2.4GHz Intel(R) Xeon(R) E5530 CPU, 2 processors (16 cores), 48G memory, and
	an operating system of Windows Server 2008 R2 Enterprise (64 bits).
The code is written in C++ and compiled by Visual Studio 2010. 


We test these algorithms on the Flixster and Arnetminer datasets as we described in Section~\ref{sec:obs},
	which have the advantage that
	the influence probabilities of all edges on all topics are learned from real action trace data
	or node topic distribution data.
To further test the scalability of different algorithms, we use a larger network data DBLP, which is also used
	in \cite{wang2012scalable}.
DBLP is an academic collaboration network extracted from the online service (www.DBLP.org), where nodes 
	represent authors and edges represent coauthoring relationships.
It contains 650K nodes and 2 million edges.
As DBLP does not have influence probabilities from the real data, 
	we simulate two topics according to the joint distribution of
	topics 1 and 2 in the Flixster and follow the practice of the TRIVALENCY model in \cite{wang2012scalable}
	to rescale it into $0.1$, $0.01$, or $0.001$, standing for strong, medium, and low influence, respectively.

In terms of topic mixtures, in practice and also supported by our data,
	 an item is usually a mixture of a small number of
	topics
	thus our tests focus on testing topic mixtures from two topics.
First, we test random samples to cover most common mixtures as follows. 
	For these three datasets, we uses 50 topic mixtures as testing samples.\footnote{50 samples is mainly to fit for the slow greedy algorithm.}
Each topic mixture is uniformly selected from all possible two topic mixtures.
Second, since we have the data of real topic mixtures in Flixster dataset, we also test
	additional cases following the same sampling technique described in Section~3.1 of \cite{aslayonline}.
	We estimate the Dirichlet distribution that maximizes the likelihood over topics learned from the data.
	After the distribution is learned, we re-sample 50 topic mixtures for testing.

In the preprocessing stage, we use two algorithms, Greedy and PMIA, to pre-compute seed sets
	for MIS and BTS, except that
	for the DBLP dataset, which is too large to run the greedy algorithm, we only run PMIA.
Algorithms MIS and BTS need to pre-select landmarks $\Lambda$.
	In our tests, we use 11 equally distant landmarks $\{0, 0.1, 0.2, \ldots, 0.9, 1\}$. Each landmarks
	can be pre-computed independently, therefore we run them on 16 cores concurrently in different processes.

We choose $k=50$ seeds in all our tests and compare the influence spread and running time of each algorithm. 
For the greedy algorithm, we use $10000$ Monte Carlo simulations. 
We also use $10000$ simulation runs and take the average to obtain the influence spread for each 
	selected seed set.

In addition, we apply offline bound and online bound to estimate influence spread of optimal solutions.
Offline bound is the influence spread of any greedy seeds multiplied by factor $1/(1-e^{-1})$.
The online bound is based on Theorem~4 in \cite{leskovec2007cost}: for any seed set $S$,
	its influence spread plus the sum of top $k$ marginal influence spread of $k$ other nodes is
	an upper bound on the optimal $k$ seed influence spread.
We use the minimum of the upper bounds among the cases of $S=\emptyset$ and $S$ being one of the
	greedy seed sets selected.


\subsection{Experiment results}

\begin{figure*}[!th]
\centering
\begin{minipage}[b]{1\textwidth}
\centering
\subfigure[Arnetminer on random samples]{
{\includegraphics[width=0.48\textwidth]{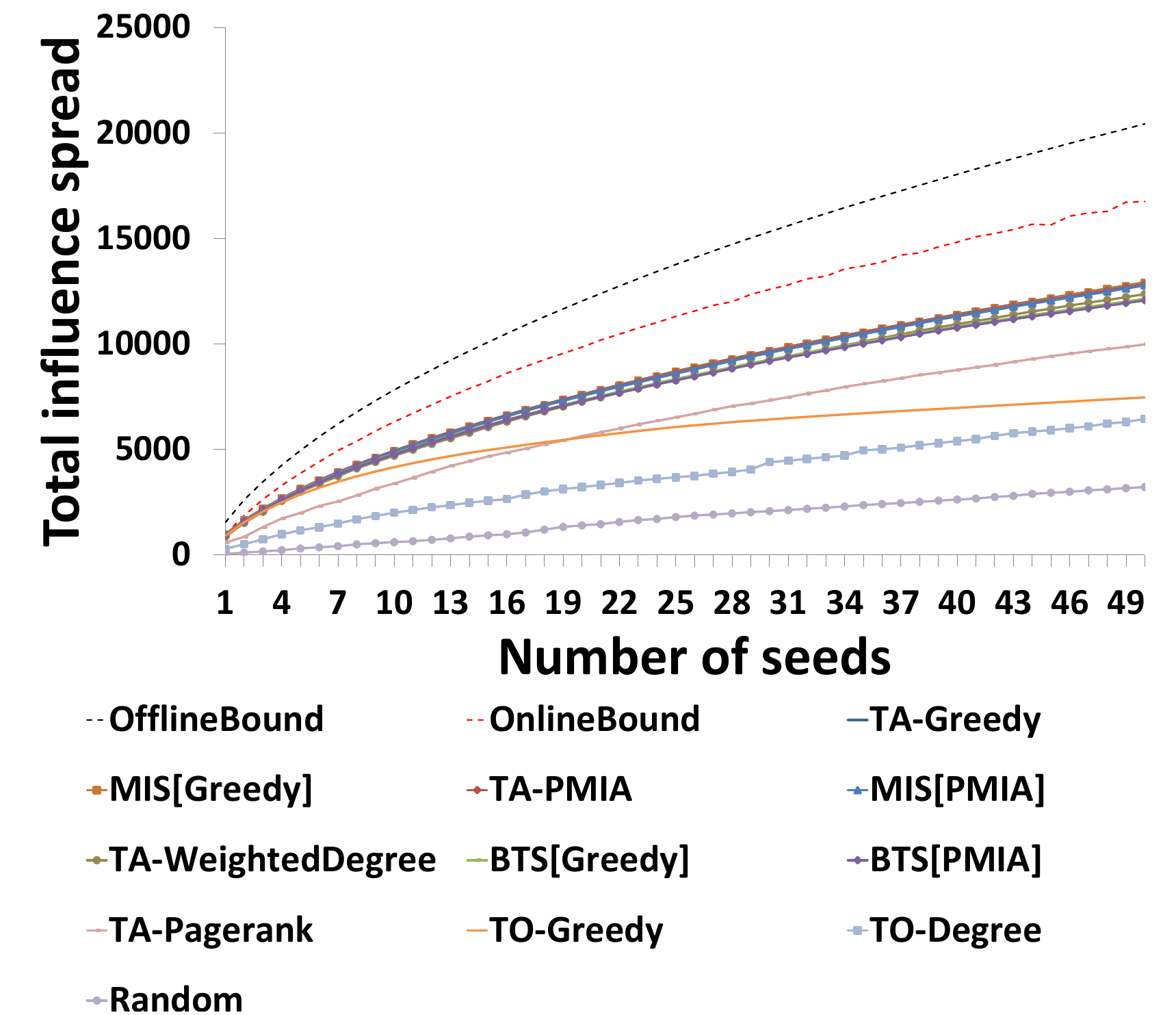}}
}
\subfigure[Flixster on random sample]{
{\includegraphics[width=0.48\textwidth]{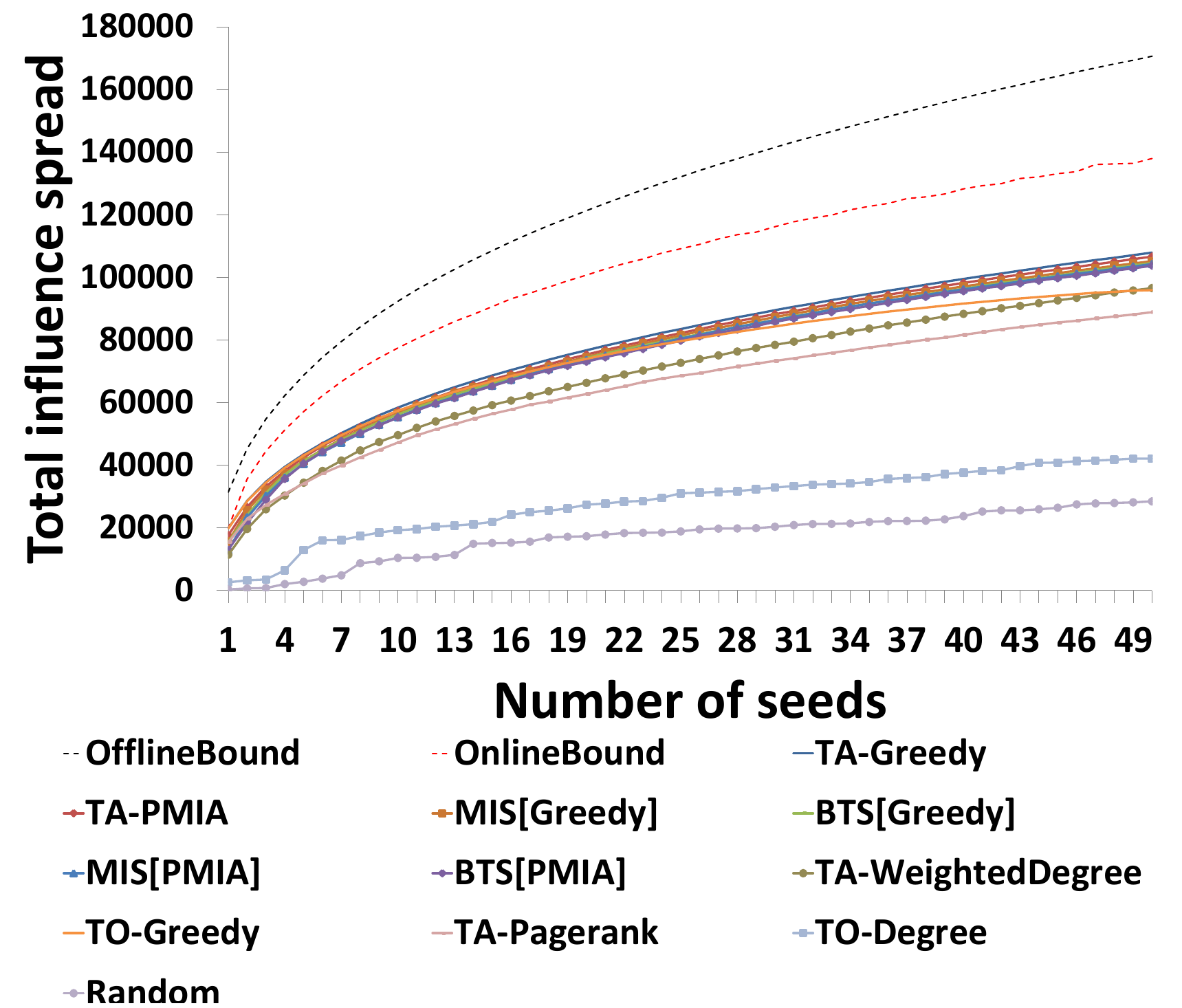}}
}
\end{minipage}
\begin{minipage}[b]{1\textwidth}
\subfigure[Flixster on Dirichlet sample]{
{\includegraphics[width=0.48\textwidth]{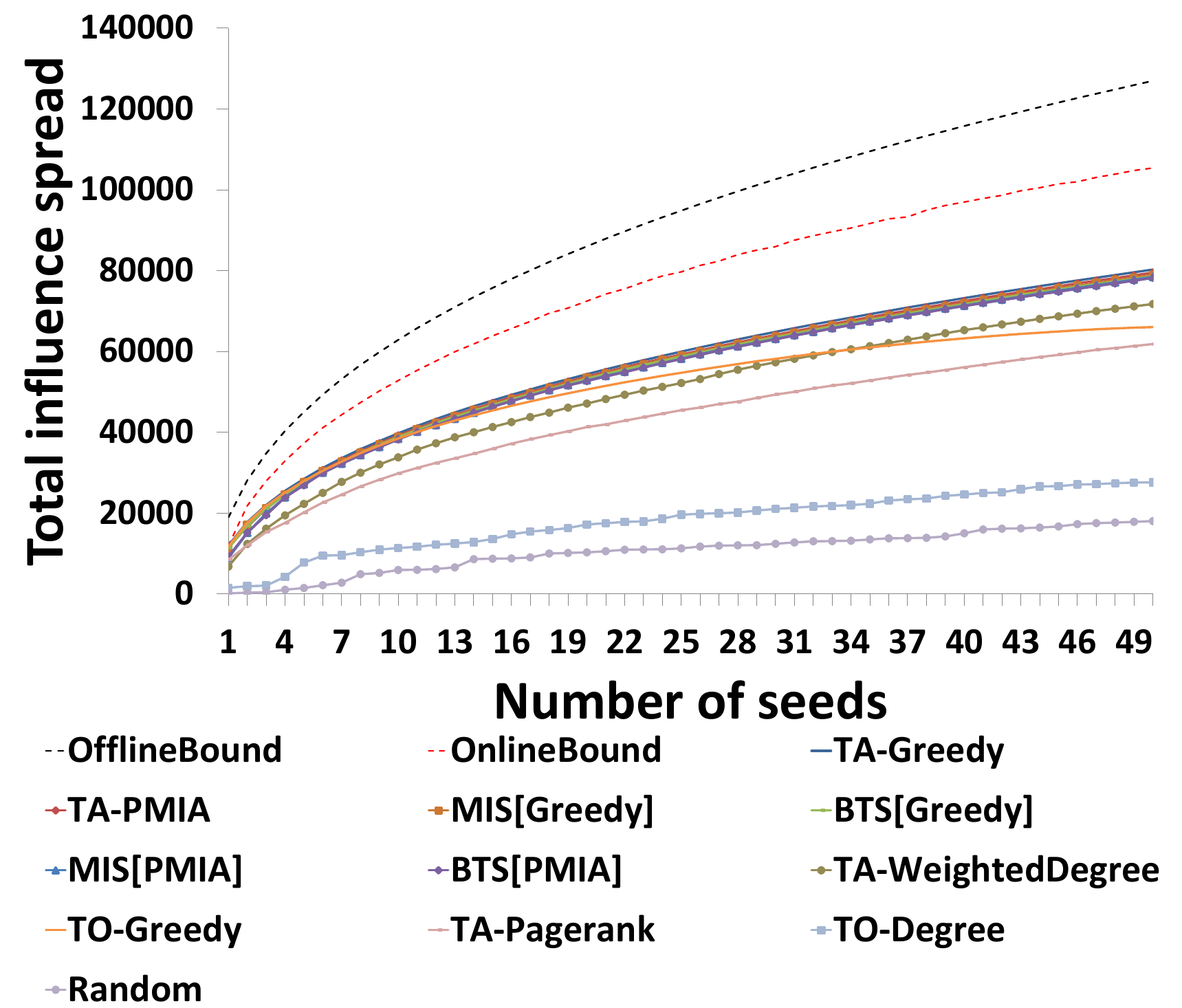}}
}
\subfigure[DBLP on random sample]{
{\includegraphics[width=0.48\textwidth]{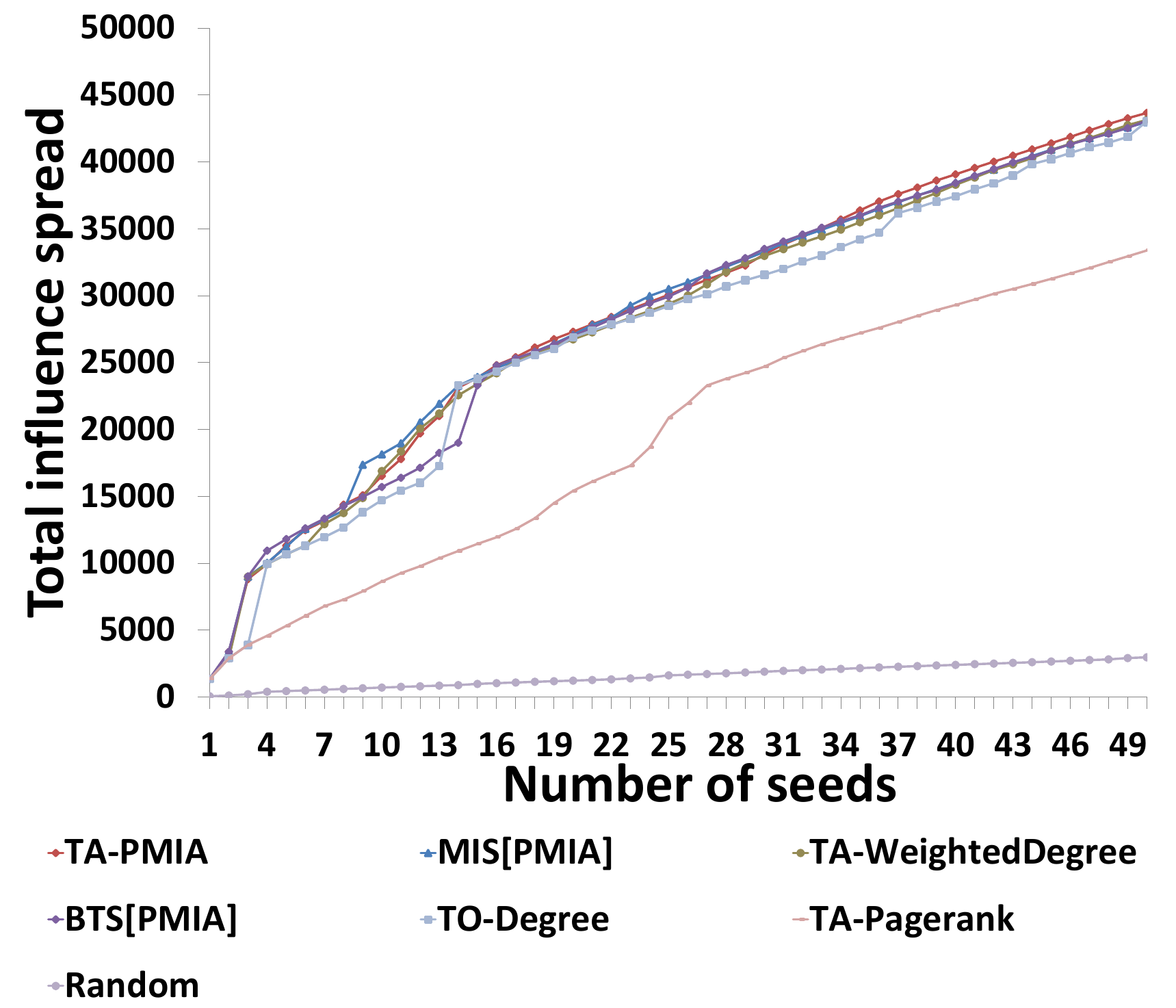}}
}
\end{minipage}

\centering
\caption{Influence spread of algorithms. Legends are ordered (left to right, top to bottom) according to influence spread.}
\label{fig:spread}
\end{figure*}

%
%
%

\begin{table}[t]
	\centering
  	\caption{Preprocessing time} \label{tab:pre_time}
  	{\tiny
  		\begin{tabular}{lcccccc}
            \toprule
                  & \multicolumn{2}{c}{Arnetminer} & \multicolumn{2}{c}{Flixster} & \multicolumn{2}{c}{DBLP} \\
                  & \multicolumn{2}{c}{{\tiny ($8\times 11$ landmarks)}} & \multicolumn{2}{c}{{\tiny ($10\times 11$ landmarks)}} & \multicolumn{2}{c}{{\tiny ($2\times 11$ landmarks)}} \\
                  & Total & Max   & Total & Max   & Total & Max \\
            \midrule
            Greedy & 8.8 hrs & 1.2 hrs & 26.3 days & 3.5 days & $\geq 100$ days & $\geq 7$ days \\
            PMIA  & 37 secs & 7.1 secs & 2.28 hrs & 12.6 mins & 9.6 mins & 4.2 mins \\
            \bottomrule
        \end{tabular}%
  	}
\end{table}%

\begin{table}[t]
	\centering
  	\caption{Average online response time} \label{tab:run_time}
  	{\scriptsize
  	\begin{tabular}{rcccc}
  	    \toprule
  	          & \multirow{2}[0]{*}{Arnetminer} & \multicolumn{2}{c}{Flixster} & \multirow{2}[0]{*}{DBLP} \\
  	          &       & random & Dirichlet &  \\
  	    \midrule
  	    \multicolumn{1}{l}{\footnotesize \textsf{TA-Greedy}} & 9.3 mins & 1.5 days & 20 hrs & N/A \\
  	    \multicolumn{1}{l}{\footnotesize \textsf{TA-PMIA}} & 0.52 sec & 5.5 mins & 3.8 mins & 58 secs \\
  	    \multicolumn{1}{l}{\footnotesize \textsf{MIS}} & 2.85 \textmu s & 2.37 \textmu s & 3.84 \textmu s & 2.09 \textmu s \\
  	    \multicolumn{1}{l}{\footnotesize \textsf{BTS}} & 1.20 \textmu s & 2.35 \textmu s & 1.42 \textmu s & 0.49 \textmu s \\
  	    \multicolumn{1}{l}{\footnotesize \textsf{TA-PageRank}} & 0.15 sec & 2.08 secs & 2.30 secs & 41 secs \\
  	    \multicolumn{1}{l}{\footnotesize \textsf{TA-WeightedDegree}} & 8.5 ms & 29.9 ms & 30.7 ms & 0.32 sec \\
  	    \bottomrule
  	    \end{tabular}%
  	}
\end{table}%



Figure~\ref{fig:spread} shows the total influence spread results on Arnetminer with random samples (a);
	Flixster with random and Dirichlet samples, (b) and (c), respectively; and DBLP with random samples (d).
Table~\ref{tab:pre_time} shows the preprocessing time based on greedy algorithm and PMIA 
	algorithm on three datasets. 
Table~\ref{tab:run_time} shows the average online response time of various algorithms
	in finding 50 seeds (topic-oblivious algorithms always use the same
		seeds and thus are not reported).

As is shown in Table~\ref{tab:pre_time},
	we run each landmark concurrently, and count both the total CPU time and the maximum time needed for one landmark.
While the total time shows the cumulative preprocessing effort, the maximum time shows the latency when we
	use parallel preprocessing on multiple cores.
The results indicate that the greedy algorithm is suitable for small graphs but infeasible for large graphs
	like DBLP, while PMIA is a scalable preprocessing solution on large graphs.
For this reason, we test two preprocessing techniques and also compare their performance.

For the Arnetminer dataset (Figure~\ref{fig:spread} (a)), it clearly separates all algorithms into three tiers:
the top tier is {\TAGreedy}, {\TAPMIA}, {\MISGreedy} and {\MISPMIA}; the middle tier is {\TAWeightDegree}, {\BTSGreedy}, {\BTSPMIA} and {\TAPageRank};
and the lower tier is topic-oblivious algorithms {\TOGreedy}, {\TODegree} and {\TORandom}.
In particular, we measure the gaps of influence spread among different algorithms.
We observe that the gap of top tiers are negligible, 
	because {\TAPMIA}, {\MISGreedy} and {\MISPMIA} 
	are only $0.61\%$, $0.32\%$ and $1.08\%$ smaller than {\TAGreedy}, respectively;
	the middle tier algorithms {\BTSGreedy}, {\BTSPMIA}, {\TAWeightDegree} and {\TAPageRank} are $4.06\%$, $4.68\%$, $4.67\%$ and $26.84\%$ smaller, respectively; 
	and the lower tier {\TOGreedy}, {\TODegree} and {\TORandom} have difference of $28.57\%$, $56.75\%$ and $81.48\%$, respectively. 
	(All percentages reported in this section are averages over 
		influence spread from one seed to 50 seeds.)

The detailed analyses are listed as follows: First, topic-oblivious algorithms does not perform well 
	in topic-aware environment.
Based on Observation~\ref{obv:topic}, when topics are separated, 
	algorithms ignoring topic mixtures cannot find influential seeds
	for all topics, and thus do not have good influence spread.
Second, {\MISGreedy} and {\MISPMIA} almost match the influence spread 
	of those of {\TAGreedy} and {\TAPMIA}.
As is indicated from offline and online bounds, {\MISGreedy}, {\BTSGreedy} are 76.9\% and 72.5\%  	
	($> 1-e^{-1}$) of the online bound, which demonstrates their
	effectiveness better than their conservative theoretical bounds could 
	support.
The MIS algorithm runs super fast in online processing, 
	finishing 50 seeds selection in just a few microseconds (Table~\ref{tab:run_time}), which	
	is three orders of magnitude faster than the millisecond response time reported in
	\cite{aslayonline}, and at least three orders of magnitude faster than any other topic-aware
	algorithms.
This is because it relies on pre-computed marginal influence and only a sorting process is needed online.
Third, {\BTSGreedy} and {\BTSPMIA} are not expected to be better than {\MISGreedy} and {\MISPMIA}, since
	BTS is a baseline algorithm only selecting
	a seed set from one topic. However, due to the preprocessing stage,
	we find that it can even perform better than other simple topic-aware
	heuristic algorithms that have short online response time.
In addition, replacing the greedy algorithm with PMIA in the preprocessing stage, MIS and BTS 
	only lose $0.76\%$ and $0.62\%$ in influence spread, indicating that PMIA is a viable choice for
	preprocessing, which greatly reduces the offline preprocessing time (Table~\ref{tab:pre_time}).


What we can conclude from tests on Arnetminer is that, for networks where topics are
	well separated among nodes and edges such as in academic networks, utilizing preprocessing
	can greatly save the online processing time.
In particular, MIS algorithm is well suited for this environment achieving microsecond response time
	with very small degradation in seed quality.

For Flixster dataset (Figure~\ref{fig:spread} (b) and (c)), we see that 
	the influence spread of {\TAPMIA}, {\MISGreedy}, {\MISPMIA}, {\BTSGreedy} and {\BTSPMIA} are $1.78\%$, $3.04\%$,
	$4.58\%$, $3.89\%$ and $5.29\%$ smaller than {\TAGreedy} for random samples, and $1.41\%$, $1.94\%$, $3.37\%$, $2.31\%$ and $3.59\%$
	smaller for Dirichlet samples, respectively.
In Flixster, we can see that for networks where topics overlap with one another on nodes,
		our preprocessing based algorithms can still perform quite well.
This is because most seeds of topic mixtures are from the constituent topics (Observation~\ref{obv:seed}).
On the other hand, the influence of {\TAWeightDegree}, {\TAPageRank} and {\TOGreedy}
	will suffer a noticeable degeneration demonstrated from two curves.
In terms of online response time (Table~\ref{tab:run_time}), the result is consistent with the result
	for Arnetminer: only MIS and BTS can achieve microsecond level
	online response, and all other topic-aware algorithms need at
	least milliseconds since they at least need a ranking computation among all nodes in the graph.
In addition, {\TAPMIA} on Flixster is much slower than on Arnetminer, because both the network size and
	the computed MIA tree size are much larger, indicating that PMIA is not suitable in providing stable
	online response time.
In contrast, the response time of MIS and BTS do not change significantly among different graphs.

In DBLP (Figure~\ref{fig:spread} (d)), the graph is too large to run greedy algorithm, 
	thus we take {\TAPMIA} as the baseline algorithm to compare with other algorithms.
	For different algorithms, the influence spread is close to each other, and our results show that {\MISPMIA} has equal competitive influence spread with {\TAPMIA} ($0.44\%$ slightly larger),
	while {\BTSPMIA}, {\TAWeightDegree}, {\TODegree} and {\TAPageRank} are $1.33\%$, $1.83\%$, $6.05\%$ and $35.54\%$
	smaller than {\TAPMIA}, respectively. 
Combining Table~\ref{tab:pre_time} and Table~\ref{tab:run_time}, we find that the greedy algorithm
	is not suitable for preprocessing for large graphs, while PMIA can be used in this case.



To summarize, the greedy algorithm has the best influence spread performance, but is slow and 
	not suitable for large-scale networks or fast response time requirements.
PMIA as a fast heuristic can achieve reasonable performance in both influence spread and online
	processing time, but its response time varies significantly depending on graph size and
	influence probability parameters, and could take
	minutes or longer to complete.
Our proposed MIS emerges as a strong candidate for fast real-time processing of topic-aware influence 
	maximization task: it achieves microsecond response time, which does not depend on graph size
	or influence probability parameters, while its influence spread matches or is very close to the
	best greedy algorithm and outperforms other simple heuristics.
Furthermore, in large graphs where greedy is too slow to finish, PMIA is a viable choice for
	preprocessing, and our MIS using PMIA as the preprocessing algorithm achieves almost the same
		influence spread as MIS using the greedy algorithm for preprocessing. 
		

\section{Related Work} \label{sec:related}
Domingos and Richardson \cite{domingos2001mining,richardson2002mining} are the first 
	to study influence maximization in an algorithmic framework.
Kempe et al. \cite{kempe2003maximizing} first formulate the 
		discrete influence diffusion models including 
		the independent cascade model and linear threshold model, 
		and provide the first batch of algorithmic results on influence maximization.

A large body of work follows the framework of \cite{kempe2003maximizing}.
One line of research improves on the efficiency and scalability of influence maximization
	algorithms \cite{goyal2011celf++,chen2009efficient,wang2012scalable,goyal2011simpath}.
Others extend the diffusion models and study other related optimization problems
	(e.g., \cite{BAA11,bhagat_2012_maximizing,HeSCJ12}).
A number of studies propose machine learning methods to learn influence models and parameters
	(e.g., \cite{saito2008prediction,tang2009social,goyal2010learning}).
A few studies look into the interplay of social influence and topic distributions 
	\cite{tang2009social,liu2010mining,weng2010twitterrank,lin2011joint}.
They focus on inference of social influence from topic distributions or joint inference
	of influence diffusion and topic distributions.
They do not provide a dynamic topic-aware influence diffusion model nor
	study the influence maximization problem.
Barbieri et al. \cite{barbieri2013topic} introduce the topic-aware influence diffusion models TIC and TLT 
	as extensions to the IC and LT models.
They provide maximum-likelihood based learning method to learn influence parameters in these
	topic-aware models.
We use the their proposed models and their dataset with the learned parameters.

A recent independent work by Aslay et al. \cite{aslayonline} is the closest one to our work.
Their work focus on index building in the query space while we use pre-computed marginal influence
	to help guiding seed selection, and thus the two approaches are complementary.
Other differences have been listed in the introduction and will not be repeated here.

%

\section{Future Work} \label{sec:conclude}

	
One possible follow-up work is to combine the advantages of our approach and the approach in
	\cite{aslayonline} to further improve the performance.
Another direction is to study fast algorithms with stronger theoretical guarantee.
An important work is to gather more real-world datasets and conduct a thorough investigation
	on the topic-wise influence properties of different networks, similar to our preliminary investigation
	on Arnetminer and Flixster datasets. 
This could bring more insights to the interplay between topic distributions and influence diffusion,
	which could guide future algorithm design.


\section*{Acknowledgments}
We would like to thank Nicola Barbieri and Jie Tang, the authors of \cite{barbieri2013topic,tang2009social}, respectively, for providing Flixster and Arnetminer datasets.

\input{sigproc-arxiv.bbl}

\end{document}

%% file: part-macros.tex
\usepackage{amsmath,amssymb}
\usepackage{algorithm}
\usepackage{algorithmic}
\usepackage{graphicx}
\usepackage{epstopdf}
\usepackage{subfigure}
\usepackage{multirow}
\usepackage[usenames,dvipsnames]{color}
\usepackage{xcolor}
\usepackage{array,booktabs}
\usepackage{etoolbox}

\newcolumntype{P}[1]{>{\centering\arraybackslash}p{#1}}

\usepackage{textcomp}

\usepackage{hyperref}
\hypersetup{colorlinks=true, urlcolor=blue, linkcolor=blue, citecolor=blue}
\usepackage{cleveref}

\usepackage[normalem]{ulem}

\newcommand{\junk}[1]{}

\newcounter{obs}

\newtheorem{observation}[obs]{Observation}

\newcommand{\argmax}{\operatornamewithlimits{argmax}}

\newcommand{\MI}{{\it MI}}
\newcommand{\MII}{{\it MI}^{g}}
\newcommand{\Seed}{S^{*}}  
\newcommand{\Seedg}{S^{g}}		
\newcommand{\oSeed}{\overline{S}^{*}}

\newcommand{\oSeedg}{\overline{S}^{g}}
\newcommand{\uSeedg}{\underline{S}^{g}}
\newcommand{\clambda}{\lambda^{c}}
\newcommand{\ulambda}{\underline{\lambda}}
\newcommand{\olambda}{\overline{\lambda}}

\newcommand{\efilter}{\tau}
\newcommand{\vfilter}{\nu}

\newcommand{\reo}{\Upsilon^{E}}

\newcommand{\rvo}{\Upsilon^{V}}

\newcommand{\AlgFont}[1]{\textsf{\small #1}}
\newcommand{\TAGreedy}{{\AlgFont{TA-Greedy}}}
\newcommand{\TAPMIA}{\AlgFont{TA-PMIA}}

\newcommand{\MISGreedy}{\AlgFont{MIS[Greedy]}}
\newcommand{\MISPMIA}{\AlgFont{MIS[PMIA]}}
\newcommand{\BTSGreedy}{\AlgFont{BTS[Greedy]}}
\newcommand{\BTSPMIA}{\AlgFont{BTS[PMIA]}}

\newcommand{\TAPageRank}{\AlgFont{TA-PageRank}}
\newcommand{\TAWeightDegree}{\AlgFont{TA-WeightedDegree}}

\newcommand{\TOGreedy}{\AlgFont{TO-Greedy}}
\newcommand{\TODegree}{\AlgFont{TO-Degree}}
\newcommand{\TORandom}{\AlgFont{Random}}

%% file: part-abstract.tex
\begin{abstract}
%
%

Influence maximization is the task of finding a set of seed nodes in a social network such that
	the influence spread of these seed nodes based on certain influence diffusion model is maximized.
Topic-aware influence diffusion models have been recently proposed to address the issue that
	influence between a pair of users are often topic-dependent and information, ideas, innovations
	etc. being propagated in networks (referred collectively as items in this paper) are typically
	mixtures of topics.
In this paper, we focus on the topic-aware influence maximization task.
In particular, 
	we study preprocessing methods for these topics to avoid redoing influence maximization for each item
	from scratch.
We explore two preprocessing algorithms with theoretical justifications.
Our empirical results on data obtained in a couple of existing studies 
	demonstrate that one of our algorithms stands out as a strong candidate providing
	microsecond online response time and competitive influence spread, with reasonable preprocessing effort.

\end{abstract}